\documentclass[conference,a4paper]{IEEEtran}

\usepackage{hyperref}

\usepackage{amssymb}
\usepackage{amsthm}

\usepackage[utf8]{inputenc}
\usepackage[T1]{fontenc}
\usepackage{url}
\usepackage{ifthen}
\usepackage{cite}
\usepackage[cmex10]{amsmath} % Use the [cmex10] option to ensure complicance
\usepackage{graphicx}
\usepackage{mathtools}

	\usepackage{cite}
	\usepackage[ruled,vlined,linesnumbered]{algorithm2e}

\usepackage{amssymb}
\usepackage{amsthm}

\usepackage[utf8]{inputenc}\usepackage[T1]{fontenc}\usepackage{url}\usepackage{ifthen}\usepackage{cite}\usepackage[cmex10]{amsmath}

\usepackage{subfigure}
\usepackage{subcaption}

\newcommand\blfootnote[1]{%
  \begingroup
  \renewcommand\thefootnote{}\footnote{#1}%
  \addtocounter{footnote}{-1}%
  \endgroup
}
\usepackage{tikz}\usetikzlibrary{arrows}

\usepackage{graphicx}\newtheorem{theorem}{Theorem}
\newtheorem{corollary}{Corollary}

\makeatletter
\newenvironment{varsubequations}[1]
 {%
  \addtocounter{equation}{-1}%
  \begin{subequations}
  \renewcommand{\theparentequation}{#1}%
  \def\@currentlabel{#1}%
 }
 {%
  \end{subequations}\ignorespacesafterend
 }
\makeatother

\newcommand{\Expect}{{\rm I\kern-.3em E}}

\usepackage{color}

\begin{document}

\title{Adapt or Wait:\\Quality Adaptation for Cache-aided Channels}

\author{Eleftherios Lampiris, Giuseppe Caire
}

\maketitle

\begin{abstract}\blfootnote{Eleftherios and Giuseppe are with the Electrical Engineering and Computer Science Department, Technische Universit\"at Berlin, 10587 Berlin, Germany -- \{eleftherios.lampiris, caire\}@tu-berlin.de.
The work is supported by the European Research Council under the ERC grant agreement N. 789190 (project CARENET).}
This work focuses on {quality adaptation as a means to counter the effects of channel degradation in wireless, cache-aided channels.}
We design a delivery scheme which combines {coded caching}, superposition coding, and {scalable source coding}, while keeping the caching scheme oblivious to channel qualities.
By properly adjusting the quality at the degraded users we are able to satisfy all demands in a time-efficient manner.
In addition, superposition coding allows us to {serve high-rate} users with {high content quality} without subjecting them to a delay penalty {caused by users with lower rate channels}.
{
We design a communication framework that covers all possible channel rate and quality configurations and we further provide algorithms that can optimise the served quality.}
An interesting outcome of this work is that a modest quality reduction at the degraded users can counter the effects of significant channel degradation.
For example, in a $100$-user {system} with normalized cache size $1/10$ at each user, if $10$ users experience channel degradation of $60\%$ compared to the rate of the non-degraded users, we show that our transmission strategy leads to a $\thicksim85\%$ quality at the degraded users and perfect quality at the non-degraded users.

\end{abstract}

\section{Introduction}

The seminal work by Maddah-Ali and Niesen \cite{maddah2014fundamental} delved into the fundamental performance aspects of a single-link, bottleneck scenario in which a server is connected to $K$ cache-aided users. In this setting, the server has access to a library of $N$ files, and each user can store the equivalent of $M$ files, denoted as a fraction $\gamma \triangleq \frac{M}{N}$ of the library, while each user synchronously requests a single file from this library.

The placement and delivery algorithms introduced in \cite{maddah2014fundamental} were designed to enable each transmission to serve $K\gamma+1$ users simultaneously, even when users requested different files.
The delivery time, as formulated in \cite{maddah2014fundamental} and normalized with respect to file-size and link-rate, takes the following form:
\begin{equation}\label{eqMNdelay}
    T_{\text{MAN}} = \frac{ K(1-\gamma)}{1+K\gamma}.
\end{equation}

A significant contribution to this line of work is the observation that the performance, as described by \eqref{eqMNdelay}, is precisely optimal under uncoded placement, as proven in \cite{wanExactOptimalityTransIT2020,yuExactUncodedTransIT2018}.
Furthermore, even for arbitrary placement schemes, the performance is within a multiplicative factor of $2.01$ from the optimal, as proved in \cite{yuFactorOf2TransIT2019}.

\subsection*{Extended applicability of Coded Caching}

Since the seminal work in \cite{maddah2014fundamental}, the main premise behind coded caching i.e., caching-enabled multicast transmissions, has been extended and modified to cater to other settings and scenarios such as, decentralized caching \cite{maddahDecentralizedToN2015,amiriDecentralizedUnevenCachesTransComm2017}, multiple transmitters \cite{ngoScalableTransWireless2018,ShariatpanahiPhysicalLayer2019TransIT,naderializadehFundamentalTransIT,lampirisCachelessITW}, file popularity \cite{augmentingSerbetciWiOpt2020,7865913,ji2017order,ji2014order,zhang2017coded}, device-to-device communications \cite{ibrahim2020device,ibrahim2018device,DBLP:journals/corr/abs-1905-05446}, asynchronous demands \cite{ghasemiAsynchronousToN2020,audienceRetentionYangTcomm2019,lampiris2021coded}, wireless channels \cite{destounisAlphaFairTIT2020,ZFE:15,8695087}.
Furthermore, due to the potential of coded caching to reduce the communication cost of content-related information through cheap pre-fetching of content bits at the users, a great amount of effort has been placed in adapting coded caching for use in wireless environments.
Such efforts have revealed how the coded caching technique can be adapted to reduce channel feedback \cite{zhang2017fundamental,zhang2015coded,lampiris2021resolving}, boost multiple transmitter gains \cite{lampirisSubpacketizationJSAC,zhangXinterferenceManagementICC2019,serbetciTransmitterSidePopularityToN2023,lampirisSubpacketizationCsitSPAWC}, serve cache-aided and non-cache-aided users simultaneously \cite{lampirisCachelessTIT2020}.

\subsection*{The worst-user effect}

A limiting factor, though, in applying coded caching in wireless channels has been the worst-user effect.
In essence, due to the non-uniform rates among the users of a wireless channel, coded caching gains can be limited by the channel of the worst user.
For example, as shown in \cite{lampirisWorstUserIZS2019}, even a single low-capacity user can double the delay of the system.

Many works have sought to design algorithms which can ameliorate this effect by proposing tools and methods such as superposition coding \cite{zhangTopologicalISIT2017}, multiple antennas \cite{ngoScalableTransWireless2018,tolliMulticast2018ISIT,ShariatpanahiPhysicalLayer2019TransIT} and power adaption \cite{amiri2018caching}, to name a few.
Recent works \cite{lampirisWorstUserIZS2019,joudeh2021fundamental} {have derived} the fundamental limits of the single antenna degraded Broadcast Channel with caching, and have designed algorithms that achieve these fundamental limits within a multiplicative factor of $2.01$.

It remains, though, a hard fact that some channel rate configurations, {such as a small portion of the users experiencing very low rates,} could impose a delay penalty which cannot be recovered through the previously mentioned techniques.
{This motivates our work here to design algorithms that are able to adapt the file quality served at the users in order to reduce the aforementioned delay penalty.}
While quality adaptation mechanisms are well-known and widely implemented in conventional unicast transmission, such as DASH - Dynamic Adaptive Streaming over HTTP~\cite{stockhammer2011dynamic,kuaSurveyRateAdaptationDASHcommSurveys2017}, the challenge here is to create an adaptive quality algorithm that works in conjunction with the multicast transmissions created by coded caching.
At the end, our {primary} goal in this work is to characterise the delivery time of a cache-aided system where every user can -- potentially -- have a different channel rate and may receive different file quality.

\subsection*{State-of-Art}

The problem of varying quality requirements at the users has been {previously} investigated in the context of coded caching, and can be divided into two lines of work.
The first line of work \cite{yang2018coded,bayat2018spatially,hassanzadeh2015distortion,ibrahim2018coded,parrinelloStatisticalQoS2019ISIT} has focused on settings where users share the same unit-capacity link, but have different quality requirements and may potentially have different cache sizes.
In \cite{yang2018coded}, each user has a different cache storage and different quality requirement and the paper's objective is the design of a placement and delivery algorithm to minimize the delivery delay.
Similarly, in \cite{ibrahim2018coded} users are also equipped with different cache sizes, but compared to \cite{yang2018coded} the authors create an optimization problem to determine both the cache contents and the size of each cache.
Further, in \cite{hassanzadeh2015distortion} the authors consider a setting where users have different cache-sizes while the files are requested {according to} a popularity distribution.
The objective is to design the placement and delivery algorithms in such a manner that maximizes the quality received at each user.
Further, in \cite{parrinelloStatisticalQoS2019ISIT} the authors assume statistical knowledge of the requested qualities.
This allows them to design the user caches based on this statistical knowledge. 
Specifically, the paper presents a linear programming problem based on Index Coding converses \cite{yossefIndexCodingTransIT2011}, through which the authors are able to design the optimal solution achieving the minimisation of the delivery time under the worst-case request pattern.

In contrast, the work \cite{bayat2018spatially} considers a setting where users are equipped with multiple antennas and are served by multiple edge nodes.
Edge nodes communicate content at the same time and users need to recover the fundamental quality layer or both the fundamental and the refinement layer by decoding all the transmissions from a select amount of edge nodes.

{The second line of work \cite{amiri2018capacity,salehi2020coded} considers a model closer to the one treated in this paper,} where users experience degraded channels and the quality is adapted at each user with the aim to limit the effect of channel unevenness.
Work \cite{amiri2018capacity} considers a setting where users have different rates and different cache sizes.
The authors optimise the cache assignment and the cached contents at the users and employ dirty paper coding and superposition coding techniques in order to deliver content to the users. 
Similarly, the work in \cite{salehi2020coded} explores the problem of cache-aided communications under non-symmetric channel rates at the users.
{The proposed approach} divides each file into quality layers using a Multiple Descriptor Code and designs an optimisation problem which decides which user subsets should be served such as to satisfy various minimum Quality of Experience criteria, while at the same time to adhere to a delivery delay constraint.

\subsection*{Results overview \& Contributions}

The focus of this work is {the analysis of the efficient} delivery of variable quality content at cache-aided users in degraded Broadcast Channels.
Specifically, {our work here covers the following three objectives:}
\begin{enumerate}
	\item {Create a caching and delivery framework for the efficient communication of an arbitrary file quality to each user for any arbitrary set of channel strength variables.}
	\item Characterise the achievable delivery time under {an arbitrary set of user channel strengths} and an {arbitrary file quality delivered at each user} and,
	\item Optimise the file quality {delivered} at each user {under any target time constraint}. 
\end{enumerate}

The first objective is achieved by combining three techniques, i) superposition coding, ii) cache-aided multicasting and, iii) {scalable video coding}\footnote{{Scalable video coding \cite{schwarz2007overview}, or more general scalable source coding, indicates a class of lossy source coding schemes that operate according to the principle of successive refinability.
The original source signal (e.g., a video sequence) is encoded in such a way that the quality of the reproduction at the decoder increases with the length of the received bitstream.
In this way, if the transmission is cut short and only an initial segment of normalised size $Q < 1$ of the encoded bitstream is received, the receiver can reproduce the source with quality which is an increasing function of $Q$, and where maximum quality is achieved for $Q = 1$, i.e. when the whole bitstream is received.}}.
{The placement and delivery algorithms are presented in Sec.~\ref{secPlacementDelivery}.}

{The second objective is satisfied} through our two main theorems. 
Th.~\ref{theGeneralCase} presents the delivery time for the {general setting where each user may have arbitrary channel strength $\alpha_{k}\in(0,1]$ and is served with file quality $Q_{k}\in(0,1]$.}
Subsequently, Th.~\ref{theMain} presents {a special case of Th.~\ref{theGeneralCase} allowing us to get a clearer insight into how quality adaptation can affect the delivery time. Specifically, this setting considers two groups of}  users, {where} the first group has perfect channel rate and thus perfect quality, while the second group has reduced channel rate and receives its files with reduced quality.

The final {objective is to optimise the served file quality at each user such that the system achieves the desired delivery time}.
Observing the two theorems we see that while in the two-type case the file quality served at the degraded users can be easily adjusted in order to achieve the desired delivery time,
{ in contrast, for the general case this quality can be allocated in many different ways}.
Thus, in this step we design quality allocation algorithms {targeting the maximisation of different metrics of interest}.
{The algorithms that we consider are presented and analysed in Sec.~\ref{secAlgosQualityAllocation} and in short are:}
\begin{itemize}
	\item \textit{Proportional fairness optimisation}, where each user's quality is proportional to the user's channel degradation,
	\item \textit{Max Min optimisation}, where the goal is to {optimise the} minimum quality guarantee.
	\item \textit{Sum quality maximisation}, where the quality allocation at each user is designed such that to maximise the overall quality.
\end{itemize}

In Fig.~\ref{QualityAllocationStrategiesComparison} we compare the outputs of these three algorithms.
An interesting observation regarding the quality allocation algorithms is that it is more ``economical'', {in the context of increasing the {overall served} quality,} to {allocate a significant portion of the file for the base layer.}
This is a consequence of the multicast nature of coded caching, where each message serves multiple users, {which means that} the base quality ``lifts'' every user, while {any other quality layer $n$ is relevant to $K-n+1$ users but still needs to be communicated via (multicast) messages which involve any subset of all $K$ users.}

An additional reason why cache-aided multicasting favours {the} increasing of the quality of the lower-rate users is tied to the results of the non-adaptive quality setting in \cite{lampirisWorstUserIZS2019} where the authors show that {in most cases} the user that forms the bottleneck is not the worst-rate user, but most likely some user in the middle and that, even if the channels of all users before the bottleneck user increased, the system performance would still be the same.
Drawing the parallel with our problem here, we see that we can significantly increase, with respect to the channel strength, the file quality of every user {with lower rate than} the bottleneck user, without impacting the delivery time.

A further observation {regarding the results of the quality optimisation algorithms} has to {do} with the increased file quality at the users {with} high rate channels.
{Due to the channels of those users being better than the bottleneck user's channel we are able to push further information and increase their file quality by taking advantage of the ``topological holes'' idea (cf.~\cite{joudeh2021fundamental}) where the authors showed that non-cacheable traffic can be transmitted along with cacheable traffic in degraded broadcast channels.
Hence, we can increase the quality at the high-rate users ``for free'' i.e., without affecting neither the delivery time nor the quality at the low-rate users.
}

\begin{figure}[th!]
  \centering
\includegraphics[width=0.9\columnwidth]{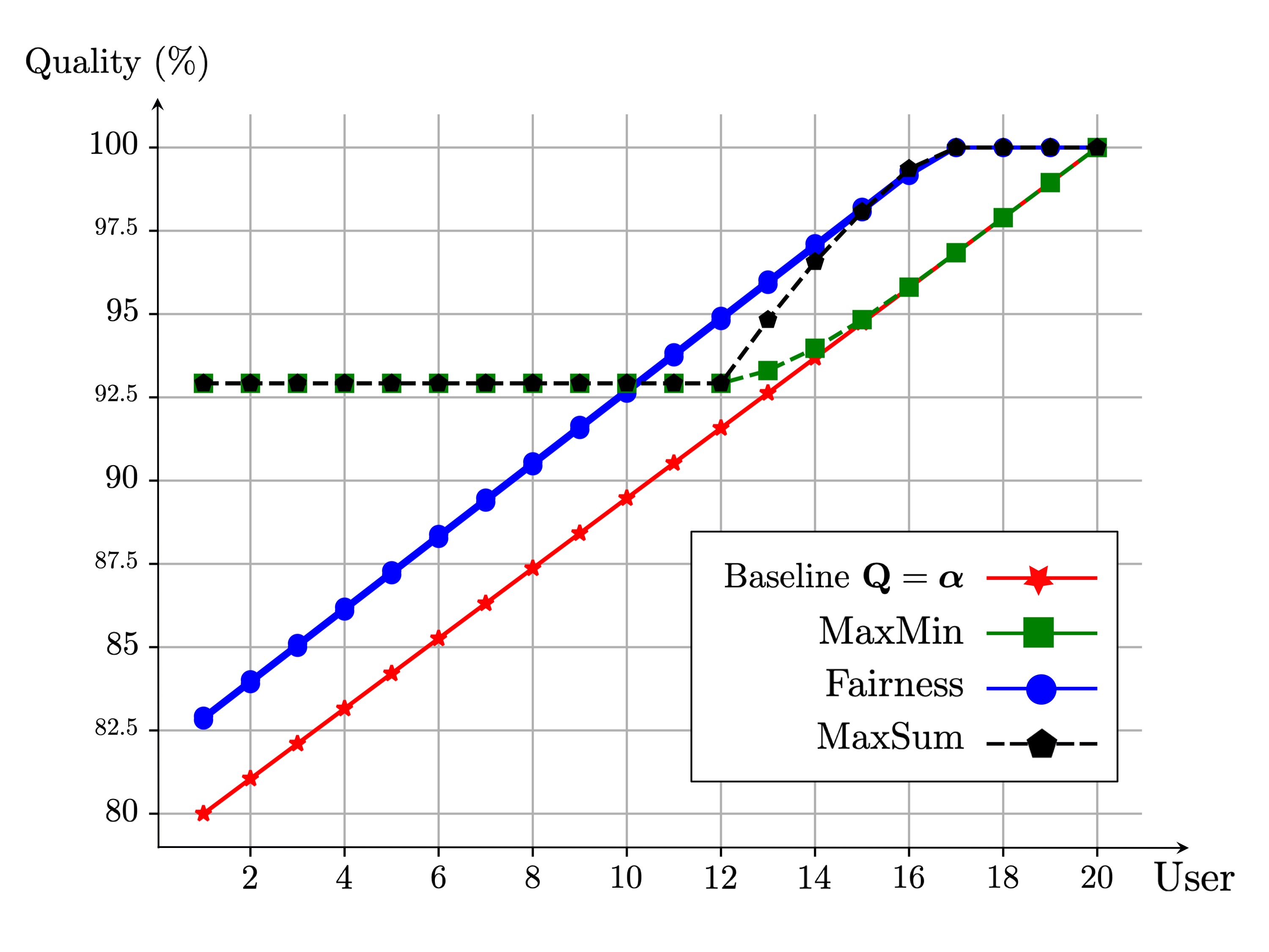}
\caption{Comparison of the quality allocation using the three afore-mentioned algorithms.
\text{Baseline} corresponding to $\mathbf{Q} \!=\! \boldsymbol\alpha$ (red solid line with stars).
\textit{Fairness optimisation} (blue solid with circles).
\textit{MaxMin optimisation} (green solid with squares).
\textit{MaxSum quality} (black dashed with pentagons).
Setting: $K\!=\!20$ users, $\gamma=\frac{3}{20}$, $\alpha_{k} = 0.8 + 0.2\frac{k-1}{K-1}$, $T_{\text{tar}}= T_{\text{MAN}}$.}
\label{QualityAllocationStrategiesComparison}
\end{figure}

\subsection*{Notation}
Symbols $\mathbb{N}, \mathbb{C}$ denote the sets of natural and complex numbers, respectively. For $n,k\in\mathbb{N},~n\geq k$, we denote the binomial coefficient with $\binom{n}{k}$, while $[k]\triangleq \{1,2, ..., k\}$.
We use $| \cdot |$ to denote the cardinality of a set. 
Throughout the paper we use the following convention for the binomial coefficient
\begin{align}
	\binom{n}{k} \triangleq 
	\begin{cases}
		\frac{ n!}{(n-k)! k!},	& n\ge k\\
		0, & n< k.
	\end{cases}
\end{align}
Additionally, throughout this work we use the Hockey-stick identity and Pascal's triangle which, respectively, take the form
\begin{align}
	\sum_{k=r}^{n}\binom{k}{m}& = \binom{n+1}{m+1} - \binom{r}{m+1}\\
	\binom{r+1}{n+1} &= \binom{r}{n} + \binom{r}{n+1}.
\end{align}

\section{Setting and System Model}

We consider the $K$-user wireless Single Input Single Output Broadcast Channel, where the transmitter has access to a library of $N$ files $\{W^n\}_{n=1}^{N}$, each of size $F$ bits, while each of the $K$ receivers is equipped with a cache of size equal to fraction $\gamma \in[0,1]$ of the total library size. Communication is divided into two distinct phases, namely the pre-fetching and the delivery phases.
{During the pre-fetching phase, which typically takes place during off-peak hours, e.g. when the devices are connected to a WiFi router}, the caches of the users are pre-filled with content from the library without any knowledge of future requests or channel capacities.
During the delivery phase, each user $k$ requests\footnote{We are interested in the worse-case delivery time and thus, we assume that each user asks for a different file.} a single file $W^{d_k}$, after which the {base station} delivers the requested content. 

A received signal {at time $t$ at user $k\in[K]$} takes the form 
\begin{align}
    y_{k}[t]= h_{k} \tilde{x}[t] +z_{k},
\end{align}
{where, $\tilde{x}[t]$ corresponds to the input signal, satisfying an average power constraint
$\frac{1}{T}\sum_{t=1}^{T} x[t]\leq P, $
$y_{k}[t]$ is the signal received at user $k$,
$h_{k}\in \mathbb{C}$ is the channel coefficient of user $k$, and $z_{k}\thicksim \mathbb{C}\mathcal{N}(0,1)$ represents the Gaussian noise at user $k$.
Under the Generalised Degrees of Freedom (GDoF) framework \cite{geng2015optimality,jafarGDoFTransIT2010,davoodiAlignedImageSetsTransIT2016,gholamiGeneralized2017TransIT}, the received signal can be re-written in its more GDoF-friendly form as follows
\begin{align}
	y_{k}[t] = \sqrt{P^{\alpha_{k}}} e^{ j \theta_{k}} x[t] + z_{k} 
\end{align}
where here $x[t] \triangleq \frac{\tilde{x}[t]}{\sqrt{P}}$ is the power-normalised transmitted signal, while $\sqrt{P^{\alpha_{k}}}$ and  $e^{ j \theta_{k}}$ are the magnitude and phase of the channel coefficient, respectively.
Further, exponent $\alpha_{k}$ is defined as the \textit{channel strength} and given by the following
\begin{align}
	\alpha_{k} \triangleq \frac{ \log( \max\{ 1, | h_{k}|^{2} \})}{\log (P)}.
\end{align}

}

{The channel strength variables in practise depend on the strength of the received signal, which in turn is a function of the pathloss which is location dependent.
We assume that for the whole duration of the delivery time the path loss of each user is known (therefore the variables $\alpha_{k}$ are known) since these depend on slowly varying user motion across the cell.}
Without loss of generality, $\alpha_{k}=1$ corresponds to the highest possible channel rate.
We assume an arbitrary set of such normalised rates $\boldsymbol{\alpha}\triangleq\{\alpha_k\}_{k=1}^{K}$ and we further assume, without loss of generality, those to be ordered in ascending order ($\alpha_{k}\le \alpha_{k+1}$).

We assume that each file $W^{n}$ is a source (e.g., video file), encoded using {(lossy)} scalable {source} coding, {and} we make the assumption that the bitstream produced by the scalable encoder can be ``cut'' at any points $q_{1}\!\cdot\! F$,..., $q_{K}\!\cdot\! F$, $q_{k}\in[0,1]$ $\forall k\in[K]$, such that the first fraction $W^{n, q_{1}}$ can be decoded at some quality that depends on $q_{1}$.
Concatenating further continuous fractions $W^{n, q_{2}}, ... ,  W^{n, q_{m}} $ would allow the decoding of the file with quality that depends on $q_{1}+...+q_{m}$, while concatenating all $K$ fractions leads to the decoding of the file with maximum quality.
Hence, user $k\!\in\![K]$ achieves quality $Q_{k}$ by decoding all layers $q_{1},\!\dots\!, q_{k}$.
Since a user with a higher capacity channel is also able to decode the signals of lower capacity users, we assume $q_{k}\geq0$, meaning that the quality at user $k$ is lower or equal than the quality of user $k+1$, $\forall k\in[K-1]$.
For simplicity, the quality function that we use in this work is linear in the size of the delivered message. 

With some slight abuse of notation we express the quality of a file as a function of the user's rate e.g., $Q_{k}\! =\! \alpha_{k}$.
In reality, since both values are normalised, this relationship is better expressed as: \textit{If the channel strength is $\alpha_{k}\! =\! 0.6$ then the quality is $Q_{k}\! = \!0.6$.}
As we show in Cor.~\ref{corAlphaGuarantee} this comparison is important, since {choosing $Q_{k}$ such that} $Q_{k} =\alpha_{k}$ would always produce a delivery time that is less or equal to {that} of the non-degraded system.
Hence, the assignment $Q_{k} = \alpha_{k}$ can be considered the ``baseline'', and at the same time the ratio $Q_{k}/\alpha_{k}$ can give us a metric of how better the quality at each user is compared to its channel degradation.

\subsection*{A note on the {system metric}}

{

In \cite{maddah2014fundamental} the considered physical model is a symmetric channel between the transmitter and the users, hence the delivery time required to communicate all $K$ files takes the form
\begin{align}\label{eqMNdelayPhysical}
	\mathcal{T}_{\text{MAN}} = \frac{1}{\log(1+P)} \frac{K(1-\gamma)}{K\gamma+1}
\end{align}
with $P$ being the transmitted power.
The above metric represents the number of channel uses required to deliver one bit of content for each user as the file size approaches infinity.

The physical model we consider in this work, though, is based on each user experiencing (potentially) different SNR, making an exact delay difficult to analyse (see also the discussion in \cite{joudeh2021fundamental,amiri2018caching,amiri2018capacity,salman2019exact}).
For this reason, we analyse the system's performance in the more tractable GDoF regime and we use its reciprocal to be the delay metric.
This metric is also known as the Generalised Normalised Delivery Time (GNDT) \cite{piovanoRobustISIT2018,zhangTopologicalISIT2017,piovanoGDoFTIT2019,lampiris2017cache} and combines the GDoF with the amount of information required to be communicated. 

The GNDT is measured in time slots, where $1$ time slot corresponds to the delay required to deliver, in the absence of caches, a single file with full quality to the strongest user as $P$ approaches infinity, which amounts to $\log(1+P)$.
Hence, if the overall delay to communicate all $K$ files, each with quality $Q_{k}$ would be $\mathcal{T}(P, \mathbf{Q})$, then the GNDT of the system would take the form
\begin{align}
	T = \lim_{P\to \infty} \mathcal{T}(P, \mathbf{Q})\log(1+P).
\end{align}}

\section{Results}

{The results presented in this section correspond to a system with $K$ users, where each user has access to a cache of normalised size $\gamma$ and requests a single file from a library of $N$ files.
Each user experiences a channel strength equal to $\alpha_{k}\in(0,1]$, for which without loss of generality holds the relationship $\alpha_{k}\leq \alpha_{k+1}$, $\forall k\in[K-1]$.
The quality communicated to user $k$ is $Q_{k}\in(0,1]$, while the relationship $Q_{k}\leq Q_{k+1}$ holds for any $k\in[K-1]$.}

\begin{theorem} \textbf{Adaptive quality for users with different rates:}\label{theGeneralCase}
{
The achievable delivery time of the aforementioned system takes the form}
\begin{align}\label{eqTimeMultiRates}
	T(K,\gamma,\mathbf{Q}, \boldsymbol\alpha) \!=\! \max_{w\in [K]} \left\{  \frac{ Q_{w}\binom{K -1 }{K\gamma}+ \ldots+ Q_{1}\binom{K -w }{K\gamma}
 }{ \alpha_{w} \cdot \binom{K}{K\gamma}} \right\},
\end{align}
where the numerator represents the amount of information needed to be communicated to users of set $[w]$.

\end{theorem}

\begin{proof}
	The proof is constructive and described Sec.~\ref{secPlacementDelivery}.
\end{proof}

\begin{theorem}{\bf{Adaptive quality for wireless channels with two-type users:}} \label{theMain}
{In the special case of the above system, where the low-rate users have channel strength $\alpha \triangleq\alpha_{1}=...=\alpha_{w} $, and quality $Q \triangleq Q_{1}=...=Q_{w}$ while the remaining users have perfect channel strength and receive perfect quality, the achievable delivery time takes the following form}
    \begin{align}\label{eqCompletionTime}
        T(K,\gamma,{\alpha}, Q) \!=\!   \max \left\{ \frac{Q}{\alpha}\cdot   \frac{ \binom{K}{K\gamma+1}\!-\!\binom{K-w}{K\gamma+1}}{\binom{K}{K\gamma}}, \ T_{\text{MAN}} \right\}.
    \end{align}
\end{theorem}
\begin{proof}
{
	The achievability is based on the design of Th.~\ref{theGeneralCase}, substituting $Q_{k\in[w]}$ with $Q$ and $Q_{k\notin [w] }=1$.
	}
\end{proof}

\begin{figure}[th!]
  \centering
\includegraphics[width=0.9\columnwidth]{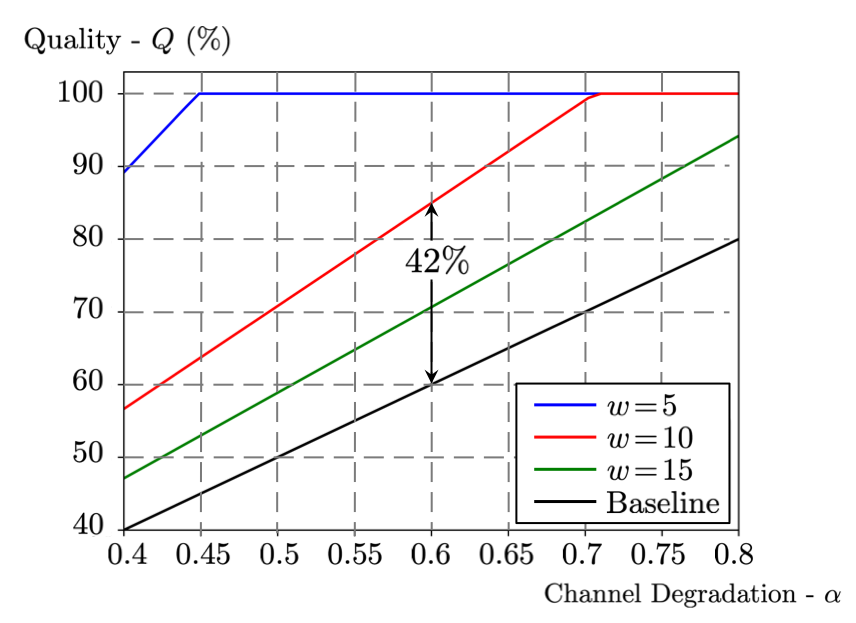}
\caption{File quality at the degraded users as a function of the channel degradation required to achieve delivery time $T_{\text{MAN}}$ for the two-type case.
Setting: $K=100$-user channel with normalized cache $\gamma=\frac{1}{10}$ at each receiver.
The baseline case corresponds to quality $Q=\alpha$.}
\label{figdegradationVquality}
\end{figure}

\begin{corollary}\label{eqCorollaryMaxQuality}
Let us consider the degraded Broadcast Channel of Th.~\ref{theMain}.
Denoting with $Q^{\star}$ the highest achievable quality that can be communicated to each user of set $[w]$ such that the delivery time equals $T_{\text{MAN}}$ of the non-degraded channel, then
\begin{equation}\label{eqQuality}
	Q^{\star}= \min \left\{  \frac{\alpha\binom{K}{K\gamma+1}}{		\binom{K}{K\gamma+1} -  \binom{K-w}{K\gamma+1}}, 1	\right\}.
\end{equation}
	
\end{corollary}
\begin{proof}
The proof is direct by equating the two parts of \eqref{eqCompletionTime} from Th.~\ref{theMain}.
\end{proof}

\begin{figure}[th!]
  \centering
\includegraphics[width=0.9\columnwidth]{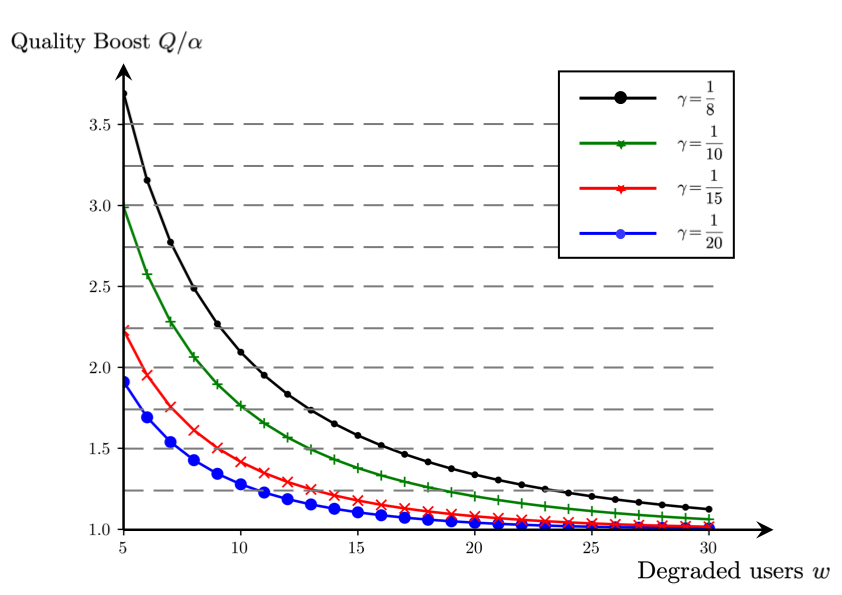}
\caption{Quality boost over the $Q \!=\! \alpha$ baseline as a function of the number of degraded users, in a channel with $K\!=\!100$ users, for various values of $\gamma$.}
\label{figBoostVusers}
\end{figure}

\begin{corollary}\label{corAlphaGuarantee}\textbf{Bound on minimum achievable quality:}
	In the setting of Th.~\ref{theGeneralCase} assigning quality at each user equal to the amount of this user's channel degradation i.e., $Q_{k}=\alpha_{k}$ would yield a delay of $T_{\text{MAN}}$ or less i.e.,
	\begin{align}
		T(K,\gamma,\mathbf{Q}\!=\!\boldsymbol\alpha, \boldsymbol\alpha) \leq T_{\text{MAN}}
	\end{align}
\end{corollary}
\begin{proof}
	The proof is presented in Appendix~\ref{appProofAlphaGuarantee}.
\end{proof}

\begin{corollary}\label{corMaxMinQuality}\textbf{Quality of Service Guarantee:}	In the setting of Th.~\ref{theGeneralCase}
with the objective of maximizing the minimum quality for target delay $T_{\text{MAN}}$, we can design the quality allocation such that
\begin{align}
	Q_{k} {=} \max\left\{ \alpha_{k}, \ \hat{Q}  \right\}, \ \forall k\in[K]
\end{align}	
where
\begin{align}
	\hat{Q} = \min_{w\in[K]} \left\{ \frac{\alpha_{w} \binom{K}{K\gamma+1} }{\binom{K}{K\gamma+1} -\binom{K-w}{K\gamma+1}}\right\}.
\end{align}
\end{corollary}
\begin{proof}
The proof is presented in Appendix~\ref{appMaxMinQuality}.
\end{proof}

\begin{corollary}\label{corFairnessGuarantee}\textbf{Proportional Fairness Guarantee:}
	In the setting of Th.~\ref{theGeneralCase} we can design a quality allocation {policy} such that the quality at every user is proportional to its channel {strength,} while achieving the delivery delay {$T_{\text{tar}} \geq T_{\text{MAN}}$} i.e.,
	\begin{align}
		Q_{k} = \min\left\{ \beta \cdot \alpha_{k}, \ 1\right\} \ \beta\geq 1, \ \forall k\in[K].
	\end{align}

\end{corollary}
\begin{proof}
	The proof is direct by {multiplying each quality $Q_{k}$ by factor $\beta\geq 1$} and using the result of Cor.~\ref{corAlphaGuarantee}.
\end{proof}

\section{Placement and Delivery algorithms}\label{secPlacementDelivery}

In this section we design and describe the placement and delivery algorithms that allow us to achieve the results of the two theorems.
The placement algorithm is borrowed from~\cite{maddah2014fundamental} and has he advantage that the pre-fetching phase requires no knowledge of future demands nor channel capacities.
{As is standard in coded caching, the pre-fetching phase is considered to bear no communication costs as it takes place during off-peak time e.g., when a user is connected to a Wi-Fi router.}

\subsection{Placement phase}

Every file $W^{n}$ is divided into $\binom{K}{K\gamma}$ subfiles, {$W^{n}_{\tau}$, where} each is denoted by a unique tuple $\tau$ of size $|\tau| = K\gamma$. The cache of user $k\in[K]$ is then filled as follows
\begin{equation}
	\mathcal{Z}_{k} = \bigcup_{\tau \ni k}  W^{n}_{\tau},  \ n\in [N].
\end{equation}

For users of set $\sigma$, $|\sigma| = K\gamma+1$ there is a unique multicast message that once communicated can satisfy part of their demands.
Denoting by $\oplus$ the bit-wise XOR operator, a multicast message for users of set $\sigma$ takes the form
\begin{align}
	X_{\sigma} = \bigoplus_{k\in\sigma} W^{d_{k}}_{\sigma \setminus\{k\}}.
\end{align}

\subsection{Delivery phase}

The delivery phase begins when all requests of the users have been communicated to the base station.
At that point the base station determines the channel capacity vector $\boldsymbol\alpha$ which, along with the target delivery time $T_{\text{tar}}\geq T_{\text{MAN}}$\footnote{{Achieving a delivery time lower than $T_{\text{MAN}}$ is of course possible by, for example, scaling each quality, even at the highest strength users, by the same factor.
Since such a scenario does not affect the design of the placement and delivery algorithms we knowingly skip such analysis.}}, allow the calculation of the quality vector $\mathbf{Q}$. 
This quality $\mathbf{Q}$ can be calculated using a plethora of metrics such as target delivery time (see Th.~\ref{theGeneralCase},~\ref{theMain}), quality of service (see Cor.~\ref{corMaxMinQuality}), maximizing the overall system quality, etc.
In other words, the objective of this step lies in finding the proper balance between adaptive quality at the users and optimising the system's resources.

At this point, though, we assume that quality vector $\mathbf{Q}$ has already been decided by the system (with one of the aforementioned criteria), and proceed with the description of the delivery algorithm.
{
By examining Th.~\ref{theGeneralCase},~\ref{theMain} we can see that the delivery time is dependent on the quality allocation $\mathbf{Q}$ that has been chosen by the base station, which we denote hereupon $T_{\text{tar}}$.}
We note that in the subsequent Sec.~\ref{secAlgosQualityAllocation} we present three different algorithms {to optimise the quality vector $\mathbf{Q}$}, and analyse their performance.

{
Each subfile $W^{n}_{\tau}$ is encoded through scalable source coding in such a manner that decoding the first fraction $Q$ of all the bits of the stream allows one to decode the file with quality $Q$.
We denote a sequence of $q_{k}$ bits of a subfile by its relative size $q_{k}$ i.e., $W^{n,q_{k}}_{\tau}, \ k\in[K]$.
And consequently, each multicast message is now comprised of subfiles of the same quality, hence a multicast message for quality $q_{k}$
 is denoted as
$X_{\sigma}^{q_{k}} = \bigoplus_{k\in\sigma} W^{d_{k},q_{k}}_{\sigma \setminus\{k\}}.$}

\subsection*{Sub-signals}

To communicate the requested files to the users we employ the {well-known} superposition coding technique\footnote{We note that superposition coding for multicasting a common message such that users with better channels decode a larger portion of the message is known as ``broadcast channel with degraded message set’’  in Information Theory (see also~\cite{el2011network}).
In addition, it is well-known that superposition coding yields a general achievable rate region for the degraded message set broadcast channel, which is optimal in the case of stochastic degradation \cite{korner1977general,nair2009capacity}, and is shown to be order optimal for the the case of coded caching in degraded broadcast channels without quality adaptation~\cite{lampirisWorstUserIZS2019,joudeh2021fundamental}.} {(cf.~\cite{cover1999elements})} where, in this case, the power of the transmitted signal $x$ is divided into $K$ sub-signals, $\{x_{k}, k\!\in\![K]\}$, each sub-signal having power $P_{k}$.
The power of every sub-signal $x_{k}$ is chosen in such a manner that it can be decodable by users $k, ..., K$.
Hence, each $x_{k}$ carries information that, along with the information of all previous sub-signals, allow user $k$ to decode its file with quality $Q_{k}$.
For example, sub-signal $x_{1}$ carries all the messages that user $1$ needs in order to successfully decode file $W^{d_{1}}$ with quality $Q_{1}$.
Sub-signal $x_{2}$ carries all the messages that user $2$ needs, apart from those that have been transmitted in $x_{1}$, in order to successfully decode file $W^{d_2}$ with quality $Q_{2}$, and so on.

\subsection*{Information amount at each sub-signal}

As discussed above, the requirement is to communicate through each sub-signal all multicast messages that have not been included in any previous sub-signal.
This means that sub-signal $x_{1}$ carries all multicast messages of quality $Q_{1}$ that include user $1$.
Further, in sub-signal $x_{2}$ we need to communicate all messages of quality $Q_1$ that include user $2$ but not user $1$ and, also, all messages with quality $Q_2$ that include user $2$.
In total, the size of the information -- measured in number of sub-files -- that each sub-signal needs to carry is
\begin{align}\label{eqTransmittedInfoA}\nonumber
	\ell_{n}(\mathbf{Q}) \!=\! q_{n}\binom{K\!-\!1}{K\gamma} &+ Q_{n\!-\!1}\binom{K\!-\!n}{K\gamma} \!+\\
		+ \! \sum_{i=2}^{n-1}&(Q_{n-1}\!-\!Q_{i-1} )\binom{K\!-\!n\!+\!i\!-\!2}{K\gamma\!-\!1}
\end{align}
where the proof being presented in Appendix~\ref{appAmountOfInfoMultiRates}.

{Further, the amount of information corresponding to all sub-signals of set $[k], \ k\in[K]$, i.e.
\begin{align}\label{eqLequalsEll}
	L_{k}(\mathbf{Q}) \triangleq \sum_{n=1}^{k} \ell_{n}(\mathbf{Q})
\end{align}
is calculated to be
}
\begin{align}\label{eqQgeqAlphaA}
	L_{k}( \mathbf{Q} ) = Q_{k}\binom{K-1}{K\gamma} + \dots + Q_{1}\binom{K-k}{K\gamma}
\end{align}
where the proof is also presented in Appendix~\ref{appAmountOfInfoMultiRates}.

We can observe that \eqref{eqQgeqAlphaA} corresponds to the numerator of Th.~\ref{theGeneralCase}.
Further, the calculation of the amount of information for the special case of two different rates is achieved by setting $Q\triangleq Q_{1}=...=Q_{w}$ and using the Hockey-stick identity, simplifying \eqref{eqQgeqAlphaA} to 
\begin{equation}
	L_{w}(Q)  = Q \cdot  \left[\binom{K}{K\gamma+1} - \binom{K-w}{K\gamma+1} \right].
\end{equation}

\subsection*{Power allocation}

The next step is the allocation of power at each sub-signal.
The main premise is that any user $k$ should be able to decode sub-signal $x_{k}$ and, by extent due to the aforementioned channel degradation hierarchy, sub-signals $x_{1}, ..., x_{k-1}$.
At the same time we seek to minimise the delivery time, subject to the power constraint.
We begin by determining set $[w]$, whose total load would produce the highest delivery time i.e.,
\begin{align}\label{eqMaxLoad}
	w = \arg\max_{k\in[K]} \left\{ \frac{L_{k}(\mathbf{Q})}{\alpha_{k}} \right\}.
\end{align}
Then, we proceed to calculate the power exponents $\pi_{n} = \frac{ L_{n} }{L_{w} } \alpha _{w}$ of each sub-signal.
By setting {$\pi_{0} =1$}, each sub-signal $x_{n}$ is eventually transmitted with power
\begin{align}
	P_{n} = P^{-\pi_{n\!-\!1}} - P^{-\pi_{n}}.
\end{align}

\subsection{Decoding at the users}

Decoding a set of superimposed signals at a receiver requires the sequential use of ``Treating Interference as Noise'' technique {(cf.~\cite{geng2015optimality})}.
At each step a user decodes the message of the highest powered sub-signal by treating each of the remaining sub-signals as gaussian noise.
Then, knowing the content of this message the user can remove it from the received signal and continue with the decoding of the next highest powered sub-signal by treating the rest as noise.
This process is repeated until all sub-signals which are received with power above the noise level have been decoded\footnote{Notice again that this successive interference cancellation at receiver k for all subsignals
$x_1, \ldots x_k$, while treating $x_{k+1}, \ldots, x_K$ as noise is well-known (and in fact, capacity achieving) for the standard (no caching) Gaussian degraded broadcast channel~\cite{cover1999elements} and order optimal for the degraded coded caching scheme without quality adaptation \cite{lampirisWorstUserIZS2019,joudeh2021fundamental}.}.

In our case, the number of sub-signals that some user $k$ can decode are $k$, due to the selected power allocation, hence this process is repeated $k$ times.
Let us examine this process at arbitrary user $k$, whose received signal takes the form
\begin{align}\label{eqSignalMultiRate}
	y_{k}\!  =\!  \sqrt{P^{\alpha_k}} ( \underbrace{x_{1}}_{ \sqrt{ 1\!-\!P^{\!-\!\pi_{1}}} } +  \underbrace{x_{2}}_{\sqrt{ P^{\!-\!\pi_{1}}\! -\! P^{-\!\pi_{2}}}}  + ... + x_{K})  +z_{k}
\end{align}
{where for simplicity we have abstained from using the phase of the channel coefficients.}
An arbitrary sub-signal $x_{n}$, ${n\leq k}$ is received at user $k$ with power $P^{\alpha_{k}} P_{n}>1$,
{
and its decoding rate is given as a function of its SNR value, taking the following form
\begin{align}
	r_{n} = \log( 1 + \text{SNR}_{n})  = \log\left(  1+ \frac{P^{\alpha_{k}}(P^{-\pi_{n-1}}-P^{-\pi_{n}})}{1 + P^{\alpha_{k}}P^{-\pi_{n}} } 	  \right).
\end{align}

Then, the GDoF rate is calculated as
\begin{align}
	R_{n} = \lim_{P\to \infty} \frac{ r_{n}}{\log(P) } = \pi_{n} - \pi_{n-1} =  \frac{\ell_{n} }{L_{w}}\alpha_{w}.
\end{align}

}

\subsection{Calculation of the delivery time}\label{secDeliveryTime}

The total delivery time is equal to the time required by any user to decode all the messages {that it needs to receive to obtain the file at the desired level of quality.}
By design, user $k$ is only interested in the first $k$ sub-signals hence, the time required to decode all messages is equal to the maximum time required to decode any of these sub-signals.

The delivery time of each sub-signal $n$ is {proportional to} its overall load, i.e. $\frac{\ell_{n}}{\binom{K}{K\gamma}}$, and {inversely proportional to} its rate, i.e. $R_{n} = \frac{\ell_{n}}{L_{w}}\alpha_{w}$.
Hence, the delivery time {$t_{n}$} for each sub-signal takes the form
\begin{align}\label{eqSubsignalTime}
	t_{n}=\frac{ \frac{\ell_{n}}{\binom{K}{K\gamma}}  }{  \frac{ \ell_{n} }{ L_{w} } \alpha_{w}} = \frac{L_{w}}{\alpha_{w}\binom{K}{K\gamma}} = T_{\text{tar}}.
\end{align}
{
which matches the target delivery time.

From \eqref{eqSubsignalTime} we conclude that the delivery of Th.~\ref{theGeneralCase} (and by extend the special case of Th.~\ref{theMain}) is achievable.

}

\section{Quality Allocation Algorithms}\label{secAlgosQualityAllocation}

\begin{figure*}[!t]
  \begin{minipage}{\textwidth}
    \centering

    \subfigure[$K\gamma = 2$]{\includegraphics[width=0.31\linewidth]{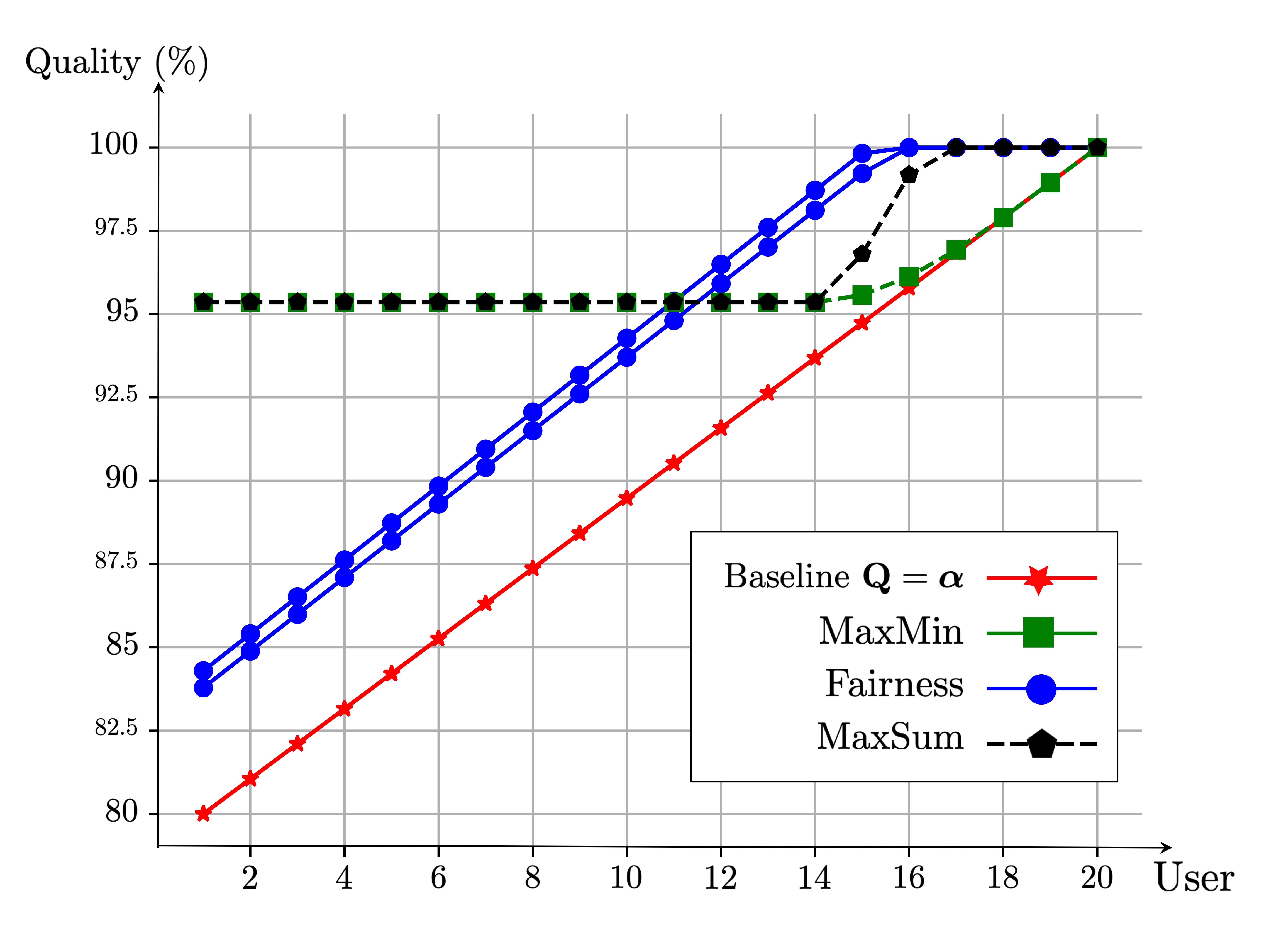}} \label{comparisonAlgos1}
    \hfill
    \subfigure[$K\gamma = 4$]{\includegraphics[width=0.31\linewidth]{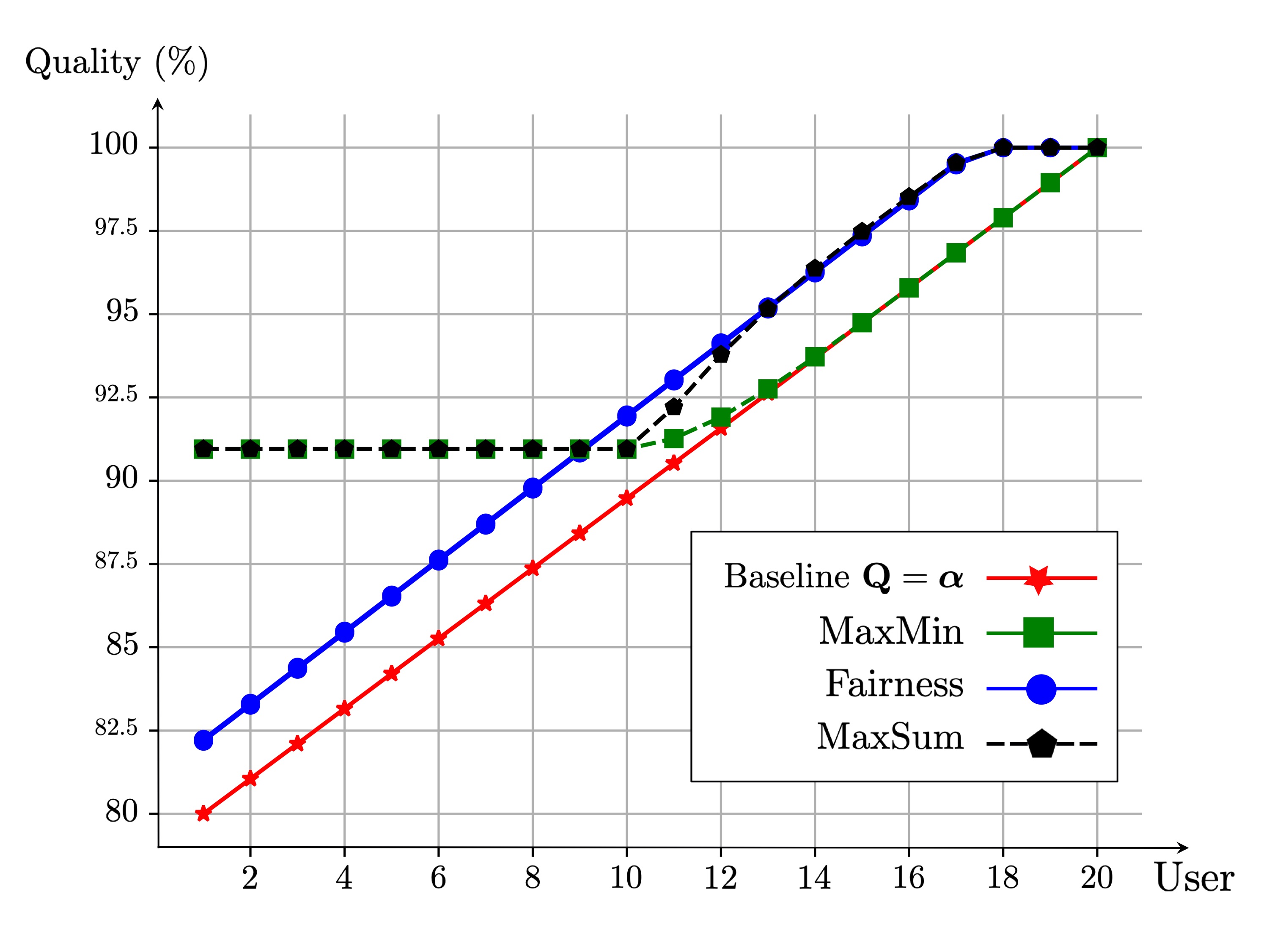}}\label{comparisonAlgos2}
    \hfill
    \subfigure[$K\gamma = 7$]{\includegraphics[width=0.31\linewidth]{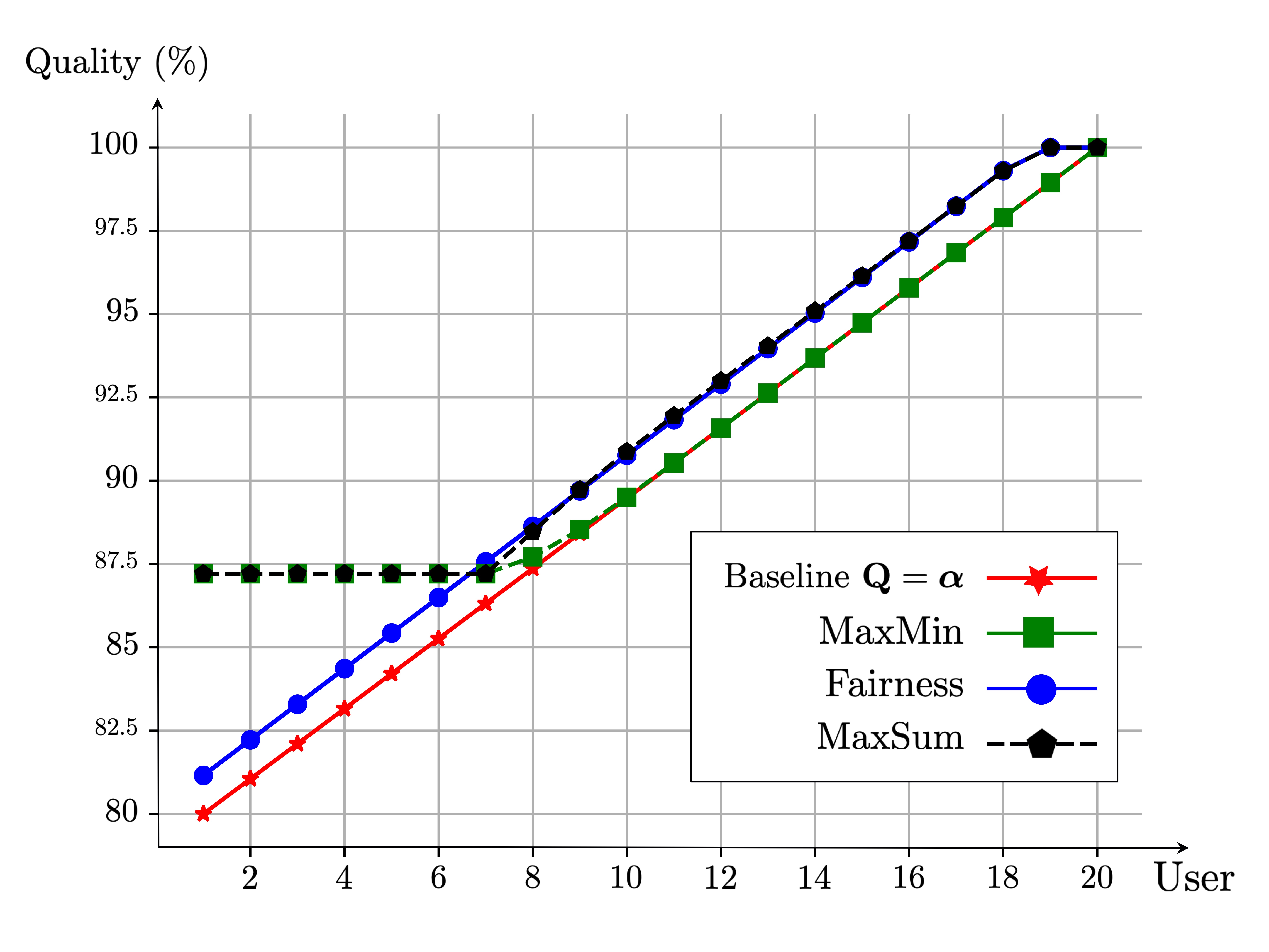}}\label{comparisonAlgos3}

    \caption{Comparison of the per-user quality achieved by each of the algorithms.
\textit{Baseline} corresponding to $\mathbf{Q} \!=\! \boldsymbol\alpha$ (red solid line with stars).
\textit{Proportional Fairness optimisation} (blue solid with circles).
\textit{MaxMin optimisation} (green solid with squares).
\textit{MaxSum quality} (black dashed with pentagons).
Setting: $K\!=\!20$ users, $\alpha_{k} = 0.8 + 0.2\frac{k-1}{K-1}$, $T_{\text{tar}}= T_{\text{MAN}}$.
The double line appearing in the fairness results in (a) is the outcome of applying the algorithm twice in order to achieve the optimal result.}
    \label{figCompThreeAlgos}
  \end{minipage}
\end{figure*}

The last objective of this work is the design of algorithms that {optimise} the quality {that is} communicated to each user such that to enable the system to achieve the target delivery time while at the same time optimising other metrics, such as fairness, and quality of service.

\subsection{Proportional Fairness Optimisation}\label{secPropFairness}

The first quality allocation algorithm that we propose is based on a proportional fairness criterion.
Specifically, the goal is to design the quality vector such that the relative quality between two degraded users to be equal to the relative channel degradation between these users.
To achieve this, we ask that the file quality at each user {to be} proportional to the user's channel {strength} i.e., $Q_{k} = \min\{ \beta\!\cdot\! \alpha_{k}, 1\}$, $\beta\geq 1$.

Cor.~\ref{corAlphaGuarantee}, along with its corresponding proof presented in Appendix~\ref{appProofAlphaGuarantee}, provide the theoretical foundation for this approach, proving that such a solution is always {attainable} for any target delivery time $T_{\text{tar}}\geq T_{\text{MAN}}$.
The optimisation problem that we seek to solve takes the form
\begin{align}
	&&&\max  \beta  \\ \tag{C1} \label{condC1PropFairness}
	& \text{s.t.} &&\frac{L_{k} (\mathbf{Q}(\beta))}{\alpha_{k}\binom{K}{K\gamma}} \leq  T_{\text{tar}} , \  k\in[K] \\
 	& && Q_{k}(\beta) = \min\{ \beta\!\cdot\! \alpha_{k}, 1\}, \  k\in [K]. \tag{C2}
\end{align}

This problem can be solved by setting $Q_{k} = \min\{ \beta\!\cdot\! \alpha_{k}, 1\}$ and solving each condition $C1$ for $\beta$.
After iterating through all the conditions we retain the highest $\beta$ that satisfies all conditions $C1$.

{
\subsubsection*{\textbf{Note}}

The above algorithm ensures that each user quality is a multiple of the channel rate of that user.
Though, in some cases, one can further improve these qualities by making the following observation.
If for some user $n\in[K]$ the channel strength is $\alpha_{n}\!<\!1$, while at the same time $\beta\! \cdot\! \alpha_{n} \!>\!1$, the resulting quality should be $Q_{n} = 1$, which is strictly less than $\beta\cdot\alpha_{n}$.
This means, though, that conditions \eqref{condC1PropFairness} are not met with equality for any $k\ge n$, which further means that the quality of some users could be higher.

We can easily overcome this, and as a result increase the quality of some users, by pre-allocating the ``maxed-out'' qualities as $Q_{n} = 1$ and continue with a recalculation of $\beta$.
This effect is depicted in Fig.~\ref{figCompThreeAlgos}-(a), where we can see a double line in the calculation of the proportional fairness qualities.
}

\subsection{Max min optimisation}

The second quality allocation algorithm {that} we propose aims to provide a guarantee on the quality of service, by maximising the worst-user quality, hence providing a minimum quality at every user.
Optimising the smallest quality is equivalent to maximising $Q_1$, and as a consequence, calculating the MaxMin quality boils down to calculating the minimum among a set of values i.e.,
\begin{align}
	Q_{\text{MaxMin}} = \min_{w\in [K]} \left\{ \frac{\alpha_{w} \ T_{\text{tar}}\ \binom{K}{K\gamma}}{\binom{K}{K\gamma+1} - \binom{K-w}{K\gamma+1}} \right\}.
\end{align}
The final quality allocation at the users takes the form
\begin{align}
	Q_{k} = \max \{  Q_{\text{MaxMin}} , \ \alpha_{k} \}
\end{align}
which is guaranteed to always be a solution {as we prove in Appendix~\ref{appMaxMinQuality}.}

{
It is important to note here that this solution is not unique. In fact, there might be solutions that achieve the same max-min quality guarantee, while at the same time improving the quality of users with higher rates.
One such example is the sum-quality maximisation algorithm which we present in the following subsection and which we show analytically that matches the max-min quality allocated to the low-rate users, while improving the quality of the high-rate users.
However, we consider this algorithm because a) it is a solution of the max-min problem, and b) it is computationally simple. 

}

\subsection{Sum-quality maximisation}

The final quality allocation algorithm {that} we propose aims to maximise the overall system quality.
As we {analytically prove} in this section, an interesting property of this algorithm is that it also maximises the minimum quality.
In other words, the quality vector calculated through this algorithm is also the maximal allocation vector for achieving the Max Min optimisation.
The optimisation problem that we seek to solve takes the form
\begin{align}
	&&&\max \sum_{i=1}^{K} Q_{i} \\  \tag{C1} \label{eqSumQmaxCon1}
	& \text{s.t.} &&\frac{ Q_{k}\binom{K-1}{K\gamma}+ \dots +Q_{1}\binom{K-1}{K\gamma}}{\alpha_{k}\binom{K}{K\gamma}} \leq  T_{\text{tar}} , \  k\in[K] \\
 	& && Q_{n} \leq Q_{n+1}, \  n\in [K\!-\!1]. \tag{C2} \label{eqSumQmaxCon2}
\end{align}

The main idea behind characterising the exact solution to this {problem} is to take advantage of constraints \eqref{eqSumQmaxCon1} and \eqref{eqSumQmaxCon2} in order to reduce the search space.
Specifically, we show that {the process should start with the maximisation of the quality at user $1$, i.e. $Q_{1}$, and progressively maximising the quality at each subsequent user, given the previously calculated qualities.
In other words, this result shows that maximising the overall served quality starts by maximising the size of the base quality layer, i.e. the quality layer that each user receives, then maximising the size of the next quality layer, i.e. the one where only users $2, ..., K$ receive and so on.
This important observation and its consequences are further discussed in Sec.~\ref{secIntuitionAlgos}.
}

Let us begin by examining any equation of constraint \eqref{eqSumQmaxCon1}.
We can see that maximising the sum is equivalent to maximising the qualities in an ascending order, because quality $Q_{1}$ is always paired with the smallest factor in every equation of constraint \eqref{eqSumQmaxCon1}.
Which further means, that by maximising $Q_{1}$ we are adding the smallest ``cost'' to each constraint in \eqref{eqSumQmaxCon1}, and thus maximising the sum.

The maximisation of $Q_{1}$, as we already calculated in the Max Min optimisation section, is achieved for
\begin{align}
	Q_{1} = \min_{w_{1}\in [K]} \left\{ \frac{\alpha_{w_{1}} \ T_{\text{tar}}\ \binom{K}{K\gamma}}{\binom{K}{K\gamma+1} - \binom{K-w_{1}}{K\gamma+1}} \right\}.
\end{align}

Now that we have access to $Q_{1}$ we can proceed with calculating the next quality, $Q_{2}$.
By following the same argument process as before, we can deduce that by maximising quality $Q_{2}$ would introduce the smallest cost in each of the equations of \eqref{eqSumQmaxCon1} compared to any other quality (apart from the already calculated $Q_1$), which means that maximising $Q_{2}$ maximises the sum.
Then, $Q_{2}$ is calculated
    	\begin{align}
		Q_{2} &= \min_{w\geq2}  \left\{	\frac{ a_{w} T_{\text{tar}} \binom{K}{K\gamma} - Q_{1} \binom{K-w}{K\gamma}}{ \sum_{i=2}^{w} \binom{K+i-w-1}{K\gamma}}	\right\}.
	\end{align}

{The iterative equation to calculate each $Q_{n}$ is given by}
    	\begin{align}
		Q_{n} &= \min_{w\in\mathcal{S}}  \left\{	\frac{ a_{w} T_{\text{tar}} \binom{K}{K\gamma} -\sum_{i=1}^{n-1} Q_{i} \binom{K+i-w-1}{K\gamma}}{ \sum_{i=n}^{w} \binom{K+i-w-1}{K\gamma}}	\right\}
	\end{align}
where $S = \{ n,..., K\} $.

\subsection{Intuition and Discussion regarding the algorithms}\label{secIntuitionAlgos}

Let us begin by observing in Fig.~\ref{figCompThreeAlgos}, \ref{figBoostVusersMultiRate} a comparison between the three quality {optimisation} algorithms. 
First, we can see that maximising the sum-quality produces a maximal vector of the max-min case, meaning that in a per-element comparison the sum-quality output is always either equal or higher than the max-min algorithm, {confirming our theoretical analysis.}

A further observation is that as $\gamma$ increases the outputs of the sum-quality and the proportional fairness algorithms tend to ``converge'' and at the end approach the baseline $\mathbf{Q} = \boldsymbol\alpha$ solution, {though we should note that this convergence comes into effect for values of $\gamma$ that are not very practical.}

The final observation is centred around the performance comparison between the sum-quality maximisation algorithm and the proportional fairness algorithm. 
Specifically, by improving the base quality (which is experienced by all users) we are increasing the per-user quality.
This insight aligns with findings from the non-adaptive quality setting discussed in \cite{lampirisWorstUserIZS2019}, where it is demonstrated that the performance bottleneck typically arises from users in the middle rather than the worst-rate user.
{This has the consequence that,} even if the channels of all users {with channels worst than that} of the bottleneck user were enhanced, the system performance would remain unchanged. Applying this to our current scenario, we can infer that the quality of users can be increased relative to channel rate up to the level of the bottleneck user without affecting the delivery time significantly.

Another noteworthy observation relates to the augmentation of the file quality for users with high-rate channels.
Specifically, we can elevate the received quality for all these users as their channels surpass the bottleneck user's channel, allowing them to receive information via low-powered signals that don't interfere with lower-rate users. This observation resonates with findings in \cite{joudeh2021fundamental}, where it's shown that extra traffic can be communicated via ``topological holes'' without compromising the performance of the cache-aided multicast messages.

\section{Examples}

\subsection{Two-type user example}

Let us begin with an example from the two-type case, where we assume a broadcast channel with $K\!=\!6$ users, each equipped with a cache of size $\gamma=\frac{2}{6}$, and channel rate at each user $\alpha_{1,2} = \frac{2}{3}$, and $\alpha_{3,4,5,6} = 1$.
The worst-case delivery time for this setting (full quality at every user) is $T_{\text{deg}} =\frac{8}{5}$ (cf.~\cite{lampirisWorstUserIZS2019}), while the non-degraded channel would have had a delivery time equal to $T_{\text{MAN}} = \frac{4}{3}$.
To reduce the delivery time of the degraded channel to equate that of the non-degraded channel, we can reduce the quality at the degraded users which, according to Cor.~\ref{eqCorollaryMaxQuality} would be
\begin{equation}
	Q = \frac{ \alpha \binom{K}{K\gamma+1}}{ \binom{K}{K\gamma+1} - \binom{K-w}{K\gamma+1} } = \frac{ 2/3 }{ 1- \frac{4}{20} } = \frac{5}{6}.
\end{equation}

The placement phase aligns with the approach described in \cite{maddah2014fundamental}.
When the delivery phase begins, the first step of the algorithm is to determine the size of the quality layers and proceed to encode each subfile according to these layers. Specifically, each subfile (one out of $\binom{6}{2}$ subfiles) is split into two parts, where the first part has relative size $q_{1}=\frac{5}{6}$ namely the ``base-quality'', and the second part has relative size $q_{2}=\frac{1}{6}$ namely the ``increased quality".

The transmission power is split into two parts, the ``high-powered'' serving messages of interest to the degraded users, and the ``low-powered'' serving messages that are of interest only to the non-degraded users.
The high-powered message $x_{h}$ is transmitted with power $P_{h} = {1-P^{-\alpha}}$ and carries multicast messages with base-quality subfiles, and which are of interest to at least a user of set $\{1,2\}$. On the other hand, the low-powered message $x_{\ell}$ is transmitted with power $ P_{\ell} = {P^{-\alpha}}$ and includes the increased quality messages which are of interest to at least one user from the group of the non-degraded users, as well as full quality messages that are of interest only to users of the group of the non-degraded users. 
In particular,
\begin{align}
	x_{h} & \leftarrow \{ X^{q_{1}}_{123}, X^{q_{1}}_{124}, X^{q_{1}}_{125}, X^{q_{1}}_{126}, X^{q_{1}}_{134}, \cdots, X^{q_{1}}_{256} \}, \\
	x_{\ell} & \leftarrow \{ X^{q_{2}}_{123}, X^{q_{2}}_{124}, \cdots , X^{q_{2}}_{256}, X_{345}, X_{346}, \cdots X_{456} \}.
\end{align}

Decoding at an arbitrary transmitter from the degraded users' set, the received message takes the form
\begin{equation}
	y_{1} = \underbrace{\sqrt{P^{\alpha}}x_{h}}_{\sqrt{P^{\alpha}}} + \underbrace{\sqrt{P^{\alpha}}x_{\ell}}_{\sqrt{P^{0}}}
\end{equation}
where, by simply treating the low-powered message as noise one can decode message $x_{h}$ with rate $r_{h} = 2/3$.

In a similar manner, the message at a non-degraded user, e.g. user $4$, takes the form
\begin{equation}\label{eqSICmessage}
	y_{4} = \underbrace{\sqrt{P} x_{h}}_{\sqrt{P- P^{1-\alpha}}} + \underbrace{\sqrt{P} x_{\ell}}_{\sqrt{P^{1-\alpha}}}.
\end{equation}
Then, since this user is interested in decoding both messages, first the user would decode the high powered message by treating interference as noise, achieving rate
\begin{equation}\label{eqSICdecodingRate}
	R_{h} = \lim_{P \to\infty} \frac{\log \frac{ P_{h}}{1+P_{\ell}}}{\log P} = {\alpha} + \mathcal{O}(1).
\end{equation}
Further, after the user has successfully decoded message $x_{h}$ can proceed to remove it from \eqref{eqSICmessage} and {in a similar manner} decode message $x_{\ell}$ with rate $R_{\ell} = 1-\alpha$.

Eventually, the delivery time is calculated as the maximum between the time required to deliver messages $x_{h}$ and the time required to deliver messages $x_{\ell}$. We get
\begin{align}
	T_{h} &= \frac{ Q}{\alpha} \frac{ \binom{ K}{K\gamma+1} - \binom{ K-w}{K\gamma+1}}{ \binom{K}{K\gamma}} = \frac{5/6}{2/3} \cdot \frac{16}{15} = \frac{4}{3} \\
	T_{\ell} &= \frac{ 1\!-\! Q}{1\!-\!\alpha} \frac{ \binom{ K}{K\gamma+1}\! -\! \binom{ w}{K\gamma+1}}{ \binom{K}{K\gamma}}  + \frac{1}{1\!-\!\alpha}\frac{\binom{K-w}{K\gamma+1}}{\binom{K}{K\gamma}}=  \frac{4}{3}.
\end{align}

\subsection{Quality allocation for a multiple-rate setting}\label{exMultiRates1}

We proceed with a second example, this time showcasing the algorithms at play for a setting where there are more than 2 channel rates.
Specifically, the system in this example is a broadcast channel with $K=6$ users, each equipped with a cache of normalised size $\gamma = \frac{1}{3}$ and where, the channel strength vector is $\boldsymbol\alpha = \{ \frac{4}{8}, \frac{5}{8}, \frac{6}{8}, \frac{7}{8}, 1, 1 \}$.

{Assuming that the selected target delivery time is $T_{\text{MAN}}$,} the first step becomes the calculation of the delivered file quality at each user.
To this end, we showcase the fairness and the sum-quality maximisation methods, where the MaxMin calculation is not included as {being straightforward.}

\begin{figure}[th!]
  \centering
\includegraphics[width=0.9\columnwidth]{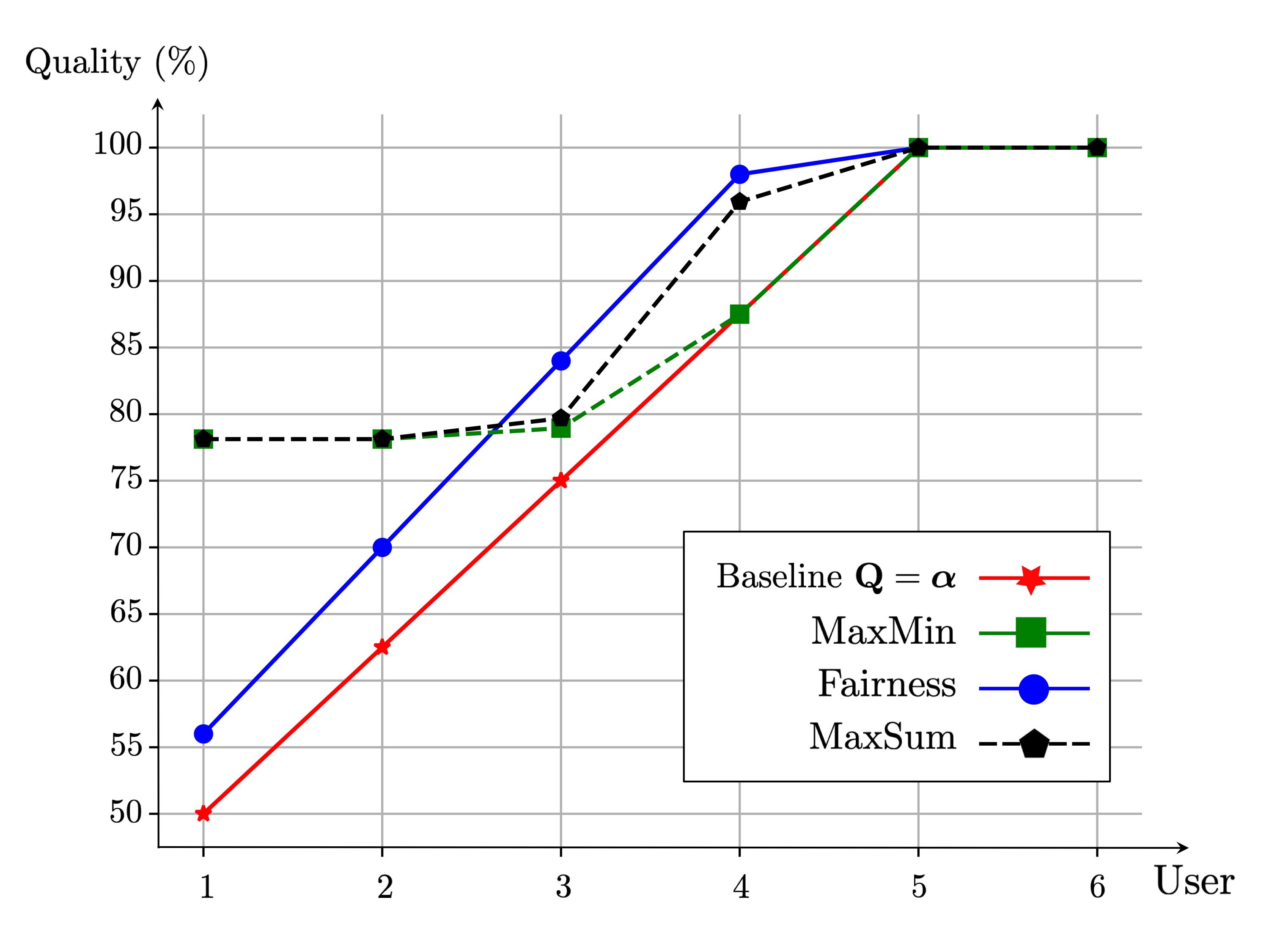}
\caption{File quality calculated by each of the three methods, for the multi-rate setting of the example presented in Sec.~\ref{exMultiRates1}.}
\label{figBoostVusersMultiRate}
\end{figure}

\subsubsection*{\textbf{Proportional Fairness optimisation}}

In order to determine the quality vector $\mathbf{Q}$ using the proportional fairness optimisation algorithm we set factor $\beta\geq 1$ such that
\begin{equation}
	\mathbf{Q} = \beta\boldsymbol\alpha =  \left( \beta\frac{4}{8},\  \beta\frac{5}{8}, \ \beta\frac{6}{8},\  \beta\frac{7}{8}, \ 1,\  1 \right).
\end{equation}

Consequently, the $6$ equations that help determine factor $\beta$, i.e. $\frac{L_{k}(\mathbf{Q})}{\alpha_{k}\binom{K}{K\gamma}} \leq T_{\text{MAN}}$, take the form
\begin{align*}
	&k=1:\ \   \beta \leq \frac{\frac{4}{8}20}{ \frac{4}{8} 10} = 2 \\
	&k=2:\ \   \beta \leq \frac{\frac{5}{8}20}{ \frac{5}{8} 10 + \frac{4}{8} 6 } = \frac{100}{74}\approx 1.35 \\
	&k=3:\ \   \beta \leq \frac{\frac{6}{8}20}{ \frac{6}{8} 10 + \frac{5}{8} 6 + \frac{4}{8} 3 } = \frac{120}{102}\approx 1.17 \\
	&k=4:\ \   \beta \leq \frac{\frac{7}{8}20}{ \frac{7}{8} 10 + \frac{6}{8} 6 + \frac{5}{8} 3+\frac{4}{8} 1 } = \frac{140}{125}= 1.12  \\
	&k=5:\ \   \beta \leq \frac{\frac{8}{8}20 -10}{ \frac{7}{8} 6+ \frac{6}{8} 3 + \frac{5}{8} 1 } = \frac{80}{65}\approx 1.23  \\
	&k=6:\ \   \beta \leq \frac{\frac{8}{8}20-10-6}{ \frac{7}{8} 3 + \frac{6}{8} 1} =\frac{32}{22} \approx 1.45 
\end{align*}
hence $\beta = \frac{140}{125}$ and the quality vector becomes $\mathbf{Q}_{\text{PFO}} = \{ \frac{14}{25}, \frac{7}{10}, \frac{21}{25}, \frac{49}{50}, 1, 1 \}$.
{We note that in this particular case the outcome of the algorithm is maximal and thus, there is no need for re-iteration (cf. note in Sec.~\ref{secPropFairness}).}

\subsubsection*{\textbf{Sum-quality maximisation}}
{Let us continue with the calculation of the quality vector in the case we are interested in maximising the sum quality among the users.
We begin by calculating the value of $Q_{1}$ as follows:}
\begin{align}
	Q_{1} = \min_{w\in[6]} \left\{ \frac{ \alpha_{w} T_{\text{MAN}}\binom{K}{K\gamma}}{ \binom{K}{K\gamma+1} - \binom{K-w}{K\gamma+1}} \right\}
\end{align}
from which we can easily see that the minimum is reached for $w = 2$, which calculates the first quality to be $Q_{1} = \frac{25}{32}$.
{
Calculating the second quality would also yield $Q_{2} = \frac{25}{32}$.}

Similarly, the calculation for $Q_{3}$ becomes
\begin{align}
	Q_{3} = \min_{w\in\{3,4,5, 6\}}  \frac{ \alpha_{w} \binom{K}{K\gamma+1}\! -\! Q_{1} \left[\binom{K\!-\!{w}}{K\gamma}\!+\!\binom{K\!-\!{w}+1}{K\gamma}\right] }{ \binom{K}{K\gamma+1}\! -\! \binom{K-2}{K\gamma+1}} 
\end{align}
from which we can easily see that the minimum is reached at $w_{2} = 3$, which means that $Q_{3} = \frac{51}{64}$.
Similarly, we continue with the calculation of the remaining qualities to arrive at the quality vector $\mathbf{Q}_{\text{m}\Sigma} = \left( \frac{25}{32},  \frac{25}{32}, \frac{51}{64}, \frac{307}{320}, 1, 1 \right)$.

\subsubsection*{\textbf{Delivery algorithm}}
The last part of this example entails the calculation of the power exponents that facilitate the delivery of the files, each with the selected quality, within the target delivery time.
Specifically, let us implement the delivery strategy using the quality vector that is the result of the Sum-Quality maximisation.
The first step is to calculate values $L_{k}(\mathbf{Q})$ and the corresponding product $L_{k}(\mathbf{Q})/\alpha_{k}$.
\begin{align}
	L (\mathbf{Q}) &= \left\{ \frac{125}{16}, \frac{25}{2}, 15, \frac{35}{2}, \frac{6057}{320}, \frac{787}{40} \right\} \\
	\ell (\mathbf{Q}) &= \left\{ \frac{125}{16}, \frac{175}{16}, \frac{5}{2}, \frac{5}{2}, \frac{457}{320}, \frac{239}{320} \right\} \\
	&\arg\max_{w\in[6]}\left\{ \frac{L_{w}(\mathbf{Q})}{\alpha_{w}} \right\} = 4.
\end{align}

Hence, the power exponents are calculated to be
\begin{align}
	\pi = \left\{	\frac{25}{64}, \frac{5}{8}, \frac{6}{8}, \frac{7}{8}, \frac{787}{800}			\right\}.
\end{align}

{The comparison between the quality achieved by each algorithm for this particular example is depicted in Fig.~\ref{figBoostVusersMultiRate}.}

{

\section{Closing remarks}

In this work we treated the problem of adaptive quality for cache-aided degraded broadcast channels.
The first outcome of this work showed how superposition coding, cache-aided multicasting and file quality adaptation can be synergistically combined to allow for a reduced communication load.
Then, we proposed a communication framework that receives as input the users' channel rates along with an (arbitrary) target file-quality at each user and provides a way to efficiently communicate a single file with different quality to each user.
Further, we provided three algorithms that optimise the file quality that needs to be served at each user based on various metrics of interest such as proportional fairness or sum-quality maximisation.

We conclude the paper with two final comments. 

\begin{enumerate}
	\item While we focused on the delivery time of a single request, it is clear that in current streaming applications a typically quite large content file (e.g. video) is pulled chunk by chunk via sequential requests.
	As already discussed in \cite{maddah2014decentralized,bayatMulticastRountingArXiv2020} one can identify as the ``files’’ of the standard Maddah-Ali and Niesen model (used in this work) the ``chunks’’ of longer content files. 
Then, the subpacketization, placement, and delivery, can be applied to all chunks of all files of the content library. This allows to handle long streaming sessions with coded caching (and in our case, with quality adaptation). 
In this case, under the condition that the delivery time is less than the playback time of the video chunks (e.g., each chunk corresponds to $2-10$ seconds of video playback) this guarantees that the streaming session of all users will not halt (i.e., the playback buffer at all video clients will remain non-empty till the end of the streaming session).  The partition of large content files into smaller chunks, and the synchronous delivery by coded caching and multicast transmission of the chunk-wise requests also allow each user to jump into the system and get out in a completely seamless way (see discussion in  \cite{bayatMulticastRountingArXiv2020}). 
	
	\item For the transmission of a single chunk (considered in this paper), it is reasonable to assume, as we did, that the user rates $\alpha_k$ are constant with time and are known at the base station.
	As usually done in virtually all modern wireless communication protocols, the base station tracks the channel quality of each user (e.g., via periodically estimating the receiver signal strength indicator), such that at any point in time the values $\alpha_k$ are known.  
In the case where the rates $\alpha_k$ change (slowly) with time along a long sequence of chunk requests (i.e., a streaming session), by repeating the 
quality allocation at any significant change of the $\alpha_k$, we can achieve 
quality adaptation. This mechanism is similar to what is done in standard 
DASH, with the fundamental difference that DASH adapts each individual user quality since it assumes unicast transmissions, while here the adaptation takes into account the fact that the underlying delivery mechanism is coded caching with multicast transmissions. 
\end{enumerate}

}

\appendices

\section{Proof of Corollary~\ref{corAlphaGuarantee}}\label{appProofAlphaGuarantee}

First, let us bound the delay for the users of set $[n]$
\begin{align}
	\frac{ L_{n}(\mathbf{Q}=\boldsymbol\alpha)}{\alpha_{n}} &=  \frac{  \alpha_{n}\binom{K-1}{K\gamma} + \dots + \alpha_{1}\binom{K-n}{K\gamma}}{\alpha_{n}}\\
	&\leq  \binom{K-1}{K\gamma} + \dots + \binom{K-n}{K\gamma}\\
	&=\! \binom{K}{K\gamma\!+\!1} \!-\! \binom{K\!-\!n}{K\gamma\!+\!1} \leq T_{\text{MAN}} \binom{K}{K\gamma}
\end{align}
where the last step used the Hockey-stick identity.
Hence,
\begin{align}
	T(\mathbf{Q}=\boldsymbol\alpha)=\max_{n\in[K]} \left\{\frac{L_{n}(\mathbf{Q}=\boldsymbol\alpha)}{\alpha_{n}\binom{K}{K\gamma}}\right\} \leq T_{\text{MAN}}\end{align}
which completes the proof.\qed

\section{Proof of Corollary~\ref{corMaxMinQuality}}\label{appMaxMinQuality}

As we discussed on the model description, each quality becomes progressively higher beginning from $Q_{1}$ due to the fact that user $n+1$ has access to at least as much information as user $n$.
Hence, the minimum quality is associated with user $1$ and which further means that the problem {can be translated to a simple maximisation of $Q_{1}$.}

Now, $Q_{1}$ has to be selected in such a way that is maximized while satisfying $T  = T_{\text{MAN}}$.
In other words, the optimization problem we seek to solve is
\begin{align}
	&\text{maximize } Q_{1} \\ \tag{\text{C}1} \label{eqMaxMinCon1}
	\text{s.t. } & \frac{Q_{n}\binom{K-1}{K\gamma}+ ... + Q_{1}\binom{K-n}{K\gamma}}{\alpha_{n}\binom{K}{K\gamma}} \leq  T_{\text{MAN}}, \ \forall n\in[K] \\
	& Q_{1}\leq Q_{2} \leq ... \leq Q_{K}. \tag{\text{C}2} \label{eqMaxMinCon2}
\end{align}

Given constraint \eqref{eqMaxMinCon2}, we can see that $Q_{1}$ is maximized when $Q_{2} = ... = Q_{n} = Q_{1} \triangleq \hat{Q}$.
Then, simply
\begin{align}\label{eqMinMaxQ}
	\hat{Q} = \min_{w\in[K]} \left\{ \frac{\alpha_{w} \binom{K}{K\gamma+1} }{\binom{K}{K\gamma+1} -\binom{K-w}{K\gamma+1}}\right\}.
\end{align}

To complete the proof, we need to show that $Q_{k} = \max\{ \alpha_{k} , \hat{Q} \}$ is an acceptable solution.
Let us begin from the quality $\hat{Q}$ of users $k\leq w$.
Then, using \eqref{eqMinMaxQ}
\begin{align}
	\hat{Q} & \leq \frac{\alpha_{n} \binom{K}{K\gamma+1}}{\binom{K}{K\gamma+1} -\binom{K-n}{K\gamma+1}},\ \forall n\in[K]\\
	\frac{\hat{Q}}{\alpha_{n}} &\leq  \frac{ \binom{K}{K\gamma+1}}{\binom{K}{K\gamma+1} -\binom{K-n}{K\gamma+1}}\leq 1 .
\end{align}

Now, for any $n > w$ we can write 
\begin{align}
	&\frac{ L_{n} ( \hat{Q}, ..., \hat{Q}, \alpha_{w+1}, ... \alpha_{n})}{\alpha_{n}} =\\
	&\!=\!\frac{\alpha_{n}\binom{K-1}{K\gamma}\!+\!... \!+\! \hat{Q}\binom{K-n}{K\gamma}}{\alpha_{n}} \!\leq\! \binom{K\!-\!1}{K\gamma}\!+\!...\!+\! \binom{K\!-\!n}{K\gamma}\\
	&=\!\binom{K}{K\gamma\!+\!1}\! -\! \binom{K-n}{K\gamma\!+\!1} \leq T_{\text{MAN}}\binom{K}{K\gamma}
\end{align}
which completes the proof.\qed

\section{Amount of information at each power layer \& Amount of information for users of set $[n]$}\label{appAmountOfInfoMultiRates}

The first step in characterizing the delay of the system is the calculation of the amount of information that needs to be transmitted at each power layer, which takes the form
	\begin{align}\label{eqTransmittedInfo2}\nonumber
		\ell_{n}(\mathbf{Q}) \!=\! q_{n}\binom{K\!-\!1}{K\gamma} &+ Q_{n\!-\!1}\binom{K\!-\!n}{K\gamma} \!+\\
		+ \! \sum_{i=2}^{n-1}&(Q_{n-1}\!-\!Q_{i-1} )\binom{K\!-\!n\!+\!i\!-\!2}{K\gamma\!-\!1}	.
	\end{align}
and as we can see consists of $3$ parts.
\begin{itemize}
	\item The first part calculates the amount of information required to communicate quality layer $n$.
%	These multicast messages are formed with mini-files of quality layer $n$ and are of interest to user $n$.
Since no user before user $n$ has interest in layer $n$ we need to communicate all quality layer $n$ multicast messages intended to user $n$, which are in total $\binom{K-1}{K\gamma}$.
	\item The second part enumerates the multicast messages that are of interest to user $n$ as well as to users $\{n+1, \dots, K\}$. 
Because no layer of any of these multicast messages has been transmitted before, we should transmit all layers $q_{1},...,q_{n-1}$ for these messages i.e., $Q_{n-1}\binom{K-n}{K\gamma}$.
	\item Finally, the third term contains remaining additional quality layers corresponding to multicast messages whose some quality layers have been transmitted in a higher power level.
	To properly calculate these, we begin by making an important observation, that if layer $k$ of a multicast message has been transmitted in a previous power level then all layers $[k\!-\!1]$ of this multicast message would be have been transmitted as well in order for user $k$ to be able to experience quality $Q_{k}$.
	Further, we can see that the second part of \eqref{eqTransmittedInfo2} contains all multicast messages with layer $1$ which have not been transmitted before.
	Using these two observations we can see that in power level $n$ one needs to transmit layers $Q_{n-1}-Q_{1}$, $Q_{n-1}\!-\!Q_{2}$, and so on, for the remaining multicast messages.
	To communicate quality $Q_{n-1}-Q_{1}$ we would need to transmit messages that include both users $1, n$ as well as $K\gamma-1$ users from $\{n+1, \dots, K\}$, which we can easily calculate those to be $\binom{K-n}{K\gamma-1}$.
	Similarly, to communicate quality $Q_{n-1}-Q_{2}$ we would need to transmit messages that include both users $2, n$ as well as $K\gamma-1$ users from $\{1, n+1, \dots, K\}$, which we can easily calculate those to be $\binom{K-n+1}{K\gamma-1}$.
	Continuing in the same manner we arrive at $Q_{n-1}-Q_{n-2}$, which messages would include both users $n-1, n$ as well as $K\gamma-1$ users from $\{1, \dots, n-2, n+1, \dots, K\}$ and which are a total of $\binom{K-3}{K\gamma-1}$ unique messages.
\end{itemize}

With this in place, let us proceed to calculate the total amount of information that needs to be communicated to the first $n$ users i.e., $L_{n}(\mathbf{Q}) = \sum_{i=1}^{n} \ell_{i}(\mathbf{Q})$.

The proof is inductive and based on showing that
\begin{align}\label{eqQgeqAlpha1}
	L_{k}( \mathbf{Q} ) = Q_{k}\binom{K-1}{K\gamma} + \dots + Q_{1}\binom{K-k}{K\gamma}.
\end{align}

Let us begin by verifying that 
\begin{align}
	L_{1} ( \mathbf{Q} )= \ell_{1}( \mathbf{Q} ) = Q_{1}\binom{K-1}{K\gamma}
\end{align}
which of course satisfies \eqref{eqQgeqAlpha1}.
Then, let us assume that \eqref{eqQgeqAlpha1} holds for some $w>1$.

We can then calculate as follows
\begin{align}
	&L_{w+1}( \mathbf{Q} ) = L_{w}( \mathbf{Q} )+\ell_{w+1}( \mathbf{Q} ) \\ \nonumber
	&=Q_{w}\binom{K-1}{K\gamma} + \dots + Q_{1}\binom{K-w}{K\gamma} + \\ \nonumber
	&+ (Q_{w+1}-Q_{w})\binom{K-1}{K\gamma} + Q_{w}\binom{K - (w+1)}{K\gamma} +\\
	&+\! (Q_{w}\!-\!Q_{1})\binom{K\!-\!(w\!+\!1)}{K\gamma-1}\! +\! \dots\! +\!(Q_{w} \!-\! Q_{w-1})\binom{K\!-\!3}{K\gamma\!-\!1}.
\end{align}

Let us first observe that the factor corresponding to $Q_{w+1}$ is equal to $\binom{K-1}{K\gamma}$.

Further, summing the factors corresponding to $Q_{ i\in[w-1]}$ becomes $\binom{K-(w-i+1)}{K\gamma} - \binom{K-(w-i+2)}{K\gamma-1}$ which, using Pascal's triangle becomes $\binom{K-(w-i+2)}{K\gamma}$.

Lastly, for the factors corresponding to $Q_{w}$ we have
\begin{align} \label{eqProofAlphaW1}
	&\binom{K\!-\!1}{K\gamma} \!-\! \binom{K\!-\!1}{K\gamma}\! +\! \binom{K\!-\!(w\!+\!1)}{K\gamma}\!+\! \sum_{i=3}^{w+1}\binom{K\!-\!i}{K\gamma\!\!-\!1} \\ \label{eqProofAlphaW2}
	&=\binom{K-(w+1)}{K\gamma}+ \binom{K-2}{K\gamma}-\binom{K-(w+1)}{K\gamma} \\ \label{eqProofAlphaW3}
	& = \binom{K-2}{K\gamma}
\end{align}
where we transformed the summation in \eqref{eqProofAlphaW1} to a difference of terms in \eqref{eqProofAlphaW2} using the Hockey-stick identity.

We showed that \eqref{eqQgeqAlpha1} holds for $n=1$ and if it holds for some $w>1$ then it also holds for $w+1$.
Hence, by induction it holds for any $n\geq 1$. \qed

\end{document}